\documentclass[graybox,vecphys]{svmult}

\usepackage{mathptmx}       
\usepackage{helvet}         
\usepackage{courier}        
\usepackage{makeidx}         
\usepackage{graphicx}        
\usepackage{multicol}        
\usepackage[bottom]{footmisc}
\usepackage{hyperref}
\usepackage{enumitem}
\usepackage{amsmath,amssymb}
\usepackage{centernot}

\newcommand{\defeq}{\triangleq}
\newcommand{\hmin}{H_{\rm{min} } }
\newcommand{\pguess}{P_{\rm{guess} } }
\newcommand{\set}[1]{\mathcal{#1}}
\newcommand{\op}[1]{\mathsf{#1}}
\newcommand{\id}{\operatorname{id}}
\newcommand{\ket}[1]{\left|#1\right\rangle}
\newcommand{\bra}[1]{\left\langle #1\right|}
\newcommand{\flag}[1]{\ket{#1}\!\!\bra{#1}}
\newcommand{\Tr}[1]{{\operatorname{Tr}}\!\left[{#1}\right]}
\newcommand{\PTr}[2]{\operatorname{Tr}_{#1}\!\left[{#2}\right]}
\newcommand{\N}[1]{\left|\!\left|{#1}\right|\!\right|}

\renewcommand{\vec}[1]{\pmb{\mathrm{#1}}}

\newcommand{\sH}{\set{H}}

\makeindex             

\begin{document}

\title*{Reverse Data-Processing Theorems\\and Computational Second Laws}
\author{Francesco Buscemi}
\institute{Francesco Buscemi \at Nagoya University, 464-8601 Nagoya Japan, \email{buscemi@is.nagoya-u.ac.jp}}
%
%
\maketitle

\abstract{Drawing on an analogy with the second law of thermodynamics for adiabatically isolated systems, Cover argued that data-processing inequalities may be seen as second laws for ``computationally isolated systems,'' namely, systems evolving without an external memory. Here we develop Cover's idea in two ways: on the one hand, we clarify its meaning and formulate it in a general framework able to describe both classical and quantum systems. On the other hand, we prove that also the reverse holds: the validity of data-processing inequalities is not only necessary, but also sufficient to conclude that a system is computationally isolated. This constitutes an information-theoretic analogue of Lieb's and Yngvason's entropy principle. We finally speculate about the possibility of employing Maxwell's demon to show that adiabaticity and memorylessness are in fact connected in a deeper way than what the formal analogy proposed here \textit{prima facie} seems to suggest.}

\section{Introduction}
\label{sec:1}

Cover, in the attempt to set the second law of thermodynamics in a computational framework, concludes his work with the following suggestive observations~\cite{cover_1996}:
\begin{quotation}\label{quote}
	The second law of thermodynamics  says that uncertainty increases in closed physical systems and that the availability of useful energy decreases. If one can make the concept of ``physical information'' meaningful, it should be possible to \textit{augment the statement of the second law of thermodynamics with the statement, ``useful information becomes less available.''} Thus the ability  of a physical system  to act as a computer should slowly degenerate as the system becomes more amorphous and closer to equilibrium. A perpetual computer should be impossible [emphasis added].
\end{quotation}
Cover's analysis can be summarized as follows. He first argues, more or less implicitly, that the computational analogue of an adiabatically isolated system should be taken to be a system evolving---i.e., computing---without an external memory. (For this reason, in what follows we use the term ``computationally isolated'' as a synonym for ``memoryless.'') This observation leads him to consider stochastic memoryless processes, in particular discrete-time Markov chains. Cover then shows that, while entropy can increase or decrease in this setting, thus violating the thermodynamical second law, \textit{relative entropy} instead never increases. We refer to this statement as Cover's ``computational second law\footnote{Since the entropy of a distribution $p$ is the \textit{negative} of the relative entropy of $p$ with respect to the uniform distribution, it is clear that Cover's computational second law formally constitutes a relaxation of the second law of thermodynamics. Indeed, the former is satisfied in situations violating the latter. We will say more about the relation between thermodynamical and computational second laws in Section~\ref{sec:app}.}.'' On a technical side, what Cover proves in~\cite{cover_1996} is an expression of the monotonicity of the relative entropy under the action of a noisy channel. Thus Cover's second law is in fact a particular \textit{data-processing inequality}~\cite{cziszar-korner,cover-thomas}, and we can imagine that there are as many computational second laws as there are data-processing inequalities, all formalizing the idea that the information content of a system cannot increase without the presence of an external memory\footnote{Relations between data-processing inequalities, the second law of thermodynamics, and statistical mechanics have been studied also by Merhav in~\cite{merhav1,merhav2,merhav3}.}.

Cover hence shows that the condition of being memoryless is sufficient for a system to obey data-processing inequalities, i.e., computational second laws. The question we address in this paper concerns the other direction: is it possible to show that the memoryless condition is also necessary for the validity of all data-processing inequalities? Equivalently stated: is it true that a system, if it \textit{is not} computationally isolated, will necessarily violate some data-processing inequality? It is important to address these questions, if we want to understand how far the analogy between memorylessness and adiabaticity can be pushed. Here, in particular, we have in mind Lieb's and Yngvason's formulation of the second law of thermodynamics~\cite{lieb-yngv}, according to which a non-decreasing entropy is not only necessary but also sufficient for the existence of an adiabatic process connecting two thermodynamical states\footnote{More on this point can be found in Section~\ref{sec:app}.}. 

The aim of this paper is to provide a comprehensive framework that is able to answer the above questions. More specifically, we prove here a family of \textit{reverse data-processing theorems}, showing that as soon as a system is not computationally isolated, it must necessarily violate a data-processing inequality. The framework we construct is quite general and it can be applied to classical, quantum, and hybrid classical/quantum systems. In fact, it may even be extended in principle to generalized operational theories as it involves only basic notions like states, effects, and operations; this development is however beyond the scope of the present work.

Thus we are able to strengthen Cover's computational second law in two ways: on the one hand, we give it a converse, in a way that is analogous to what Lieb and Yngvason did the second law of thermodynamics. On the other hand, we include in the analysis the possibility of dealing with quantum systems and quantum memories.

The paper is organized as follows. We being in Section~\ref{sec:classical} with reviewing the data-processing inequality for a classical Markov chain. This is the encoding--channel--decoding model considered by Shannon to describe the simplest communication scenario. In this scenario we prove our first \textit{reverse data-processing theorem}. We also show how this relates with the theory of comparison of noisy channels, as introduced by Shannon~\cite{shannon_1958} and later developed by K\"orner and Marton~\cite{Korner1977}. In Section~\ref{sec:lemma} we state and prove a lemma that allows us to extend our considerations to the quantum case, and discuss the notion of \textit{quantum statistical morphisms}. In Section~\ref{sec:semiclassical} we study the case of a system, processing quantum information but outputting only classical data, and prove the corresponding reverse data-processing theorem. Section~\ref{sec:quantum} presents the general case of a fully quantum computer, i.e., a process with quantum inputs and quantum outputs. Finally, in Section~\ref{sec:app}, we briefly discuss about analogies and differences between thermodynamical and computational second laws. In particular, we speculate about the possibility that Maxwell's paradox (his ``demon'') may enable a deeper relation between adiabatic processes and memoryless processes, going beyond the formal analogy considered in this work. At the end of the paper, three appendices are available: the first, reviewing conventions, notations, and terminology used in this work; the second, containing a version of the minimax theorem; and the third, presenting (just for the sake of completeness) an elementary proof of the separation theorem for convex sets.

This work contains ideas that were presented during the Sixth Nagoya Winter Workshop (NWW2015) held in Nagoya on 9-13 March 2015. Part of the technical results presented here were first introduced in previous papers by the author~\cite{buscemi-qblackwell,buscemi-all-ent,buscemi-antideg,buscemi-markov,buscemi-divisibility,buscemi-prob-inf-trans}, building upon works of Shmaya~\cite{shmaya} and Chefles~\cite{chefles}.

\section{A Reverse-Data Processing Theorem for Classical Channels}
\label{sec:classical}

A \textit{data-processing inequality} is a mathematical statement formalizing the fact that the information content of a signal cannot be increased by post-processing. As there are many ways to quantify information, so there are many corresponding data-processing inequalities. Such inequalities, however, despite formalizing the same intuitive concept, are not all \textit{logically equivalent}: some may be stronger than (i.e., imply) others, some may be easier to prove, some may be better suited for a particular problem at hand. Data-processing inequalities usually find application in information theory when proving that a given approach (coding strategy) is optimal: if a better coding were possible, that would result in the violation of one or more data-processing inequalities, thus leading to an absurd. In this sense, data-processing inequalities provide a sort of ``sanity check'' of the result.

One of the simplest scenarios in which a data-processing inequality can be formulated is the following~\cite{cziszar-korner,cover-thomas}. Given are two noisy channels $w_1:\set{X}\to\set{Y}$ and $w_2:\set{Y}\to\set{Z}$. Then, for any set $\set{U}$ and any initial joint distribution $p(x,u)$, the joint distribution $\sum_xw_2(z|y)w_1(y|x)p(x,u)$ satisfies the following inequalities:
\[
I(U;Y)\ge I(U;Z)\;.
\]
[Notations and definitions used here and in what follows are given for completeness in Appendix~1.] Referring to the situation depicted in Fig.~\ref{fig:shannon-scheme} and interpreting $U$ as the message, $X$ as the signal, $w_1$ as the communication channel, $Y$ as the output signal, $w_2$ as the decoding, and $Z$ as the recovered message, the above inequality formalizes the fact that the information content carried by the signal about the message cannot be increased by any decoding performed locally at the receiver. Of course, this does not mean that decoding should be avoided (actually, in most cases a decoding is necessary to make the signal readable to the receiver), but that no decoding is able to add \textit{a posteriori} more information to what is already carried by the signal.

\begin{figure}[t]
\centering
\includegraphics[width=0.7\linewidth]{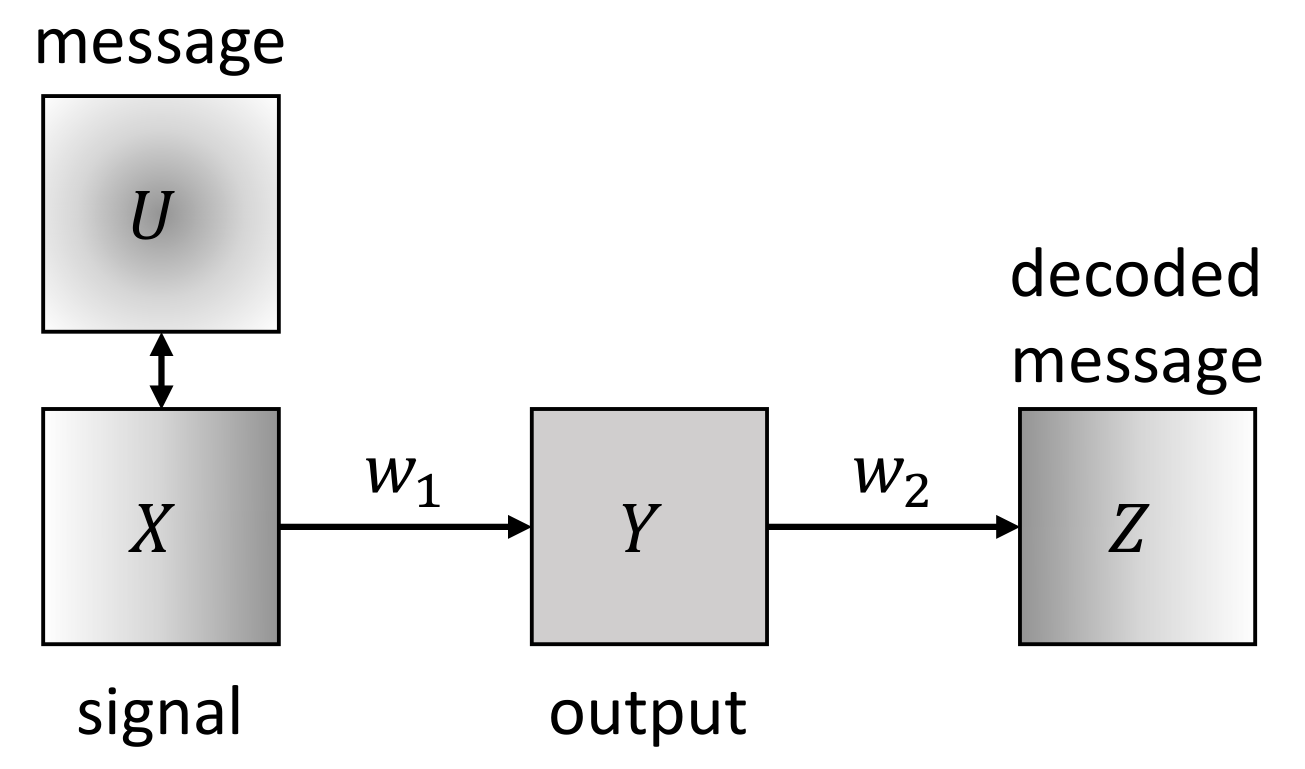}
\caption{Shannon's basic communication scheme: a message $U$ is encoded on the signal $X$ (i.e., a joint distribution $(U,X)$ is given), which is transmitted to the receiver via the communication channel $w_1$. The receiver obtains the output $Y$ and processes it according to the decoding function (another channel $w_2$) to obtain the recoveder message $Z$.}
\label{fig:shannon-scheme}
\end{figure}

Data-processing inequalities hence provide necessary conditions for the ``locality'' of the information-processing device. Namely, data-processing inequalities must be obeyed whenever the physical process carrying the message from the sender to the receiver is composed by computationally isolated parts (encoding, transmission, decoding, etc.). Any information that is communicated must be transmitted via a physical signal: as such, in the absence of an external memory, information can only decrease, never increase, along the transmission. Hence, ``locality'' in this sense can be understood as the condition that the process $U\to X\to Y\to Z$ forms a Markov chain. For this reason, we refer to such locality as ``Markov locality,'' in order to avoid confusion with other connotations of the word\footnote{In this work, memoryless process, Markov local process, and computationally isolated process are all synonyms. We prefer however to maintain all three terms because they in fact highlight different aspects of the same information-theoretic concept.}.

In this paper we aim to derive statements that provide \textit{sufficient} conditions for Markov locality, in the form of a set of information-theoretic inequalities. We refer to such statements as \textit{reverse data-processing theorems}. For example, a first attempt in this direction would be to prove the following:
\begin{quotation}
	Given are two noisy channels $w:\set{X}\to\set{Y}$ and $w':\set{X}\to\set{Z}$. Suppose that, for any set $\set{U}$ and for any initial joint distribution $p(x,u)$, the resulting distributions $\sum_xw(y|x)p(x,u)$ and $\sum_xw'(z|x)p(x,u)$ always satisfy the inequality $I(U;Y)\ge I(U;Z )$. Then there exists a noisy channel $\varphi:\set{Y}\to\set{Z}$ such that $w'(z|x)=\sum_y\varphi(z|y)w(y|x)$.
\end{quotation} 
Notice that, in the above statement, the two given channels $w$ and $w'$ are assumed to have the same input alphabet: this is a consequence of the fact that we are now formulating a \textit{reverse} data-processing theorem, so that the existence of a Markov-local decoding (the channel $\varphi$) is something to be proved, rather than being a datum. Interpreting the four random variables $(U,X,Y,Z)$ as before, if the reverse data-processing theorem holds, then we can conclude that \textit{any violation of Markov locality} is detectable, in the precise sense that the data-processing inequality has to be violated at some point along the communication process.

\subsection{Comparison of noisy channels}

A reverse data-processing theorem can be understood as a statement about the comparison of two noisy channels. Hence we want to introduce ordering relations between noisy channels, capturing the idea that one channel is able to transmit ``more information'' than another. This problem, first considered by Shannon~\cite{shannon_1958}, is intimately related to the theory of statistical comparisons~\cite{blackwell_equivalent_1953,torgersen_comparison_1991,cohen_comparisons_1998,liese-miescke}, even though this connection was not made until recently~\cite{Raginsky2011}. The theory of comparison of noisy channels received a thorough treatment by K\"orner and Marton, who in Ref.~\cite{Korner1977} introduce the following definitions (the notation used here follows~\cite{ElGamal1977}):
\begin{definition}\label{def:less-noisy}
	Given are two noisy channels, $w:\set{X}\to\set{Y}$ and $w':\set{X}\to\set{Z}$.
	\begin{enumerate}[label=\roman*)]
		\item the channel $w$ is said to be \textit{less noisy} than $w'$ if and only if, for any set $\set{U}$ and any joint distribution $p(x,u)$, the resulting distributions $\sum_x w(y|x)p(x,u)$ and $\sum_x w'(z|x)p(x,u)$ always satisfy the inequality
		\begin{equation}\label{eq:less-noisy}
			H(U|Y)\le H(U|Z)\;;
		\end{equation}
		\item the channel $w$ is said to be \textit{degradable} into $w'$ if and only if there exists another channel $\varphi:\set{Y}\to\set{Z}$ such that
		\begin{equation}\label{eq:degradable}
			w'(z|x)=\sum_y\varphi(z|y)w(y|x)\;.
		\end{equation}
	\end{enumerate}
	\qed
\end{definition}

Since $I(U;Y)\ge I(U;Z)$ if and only if $H(U|Y)\le H(U|Z)$, we immediately notice that the reverse data-processing theorem, as tentatively formulated above, is equivalent to the implication (i)$\implies$(ii): indeed, the reverse implication, (ii)$\implies$(i), is the usual data-processing inequality. K\"orner and Marton provide an explicit counterexample showing that
\begin{equation}
	\text{degradable}\underset{\centernot\impliedby}{\implies}\text{less noisy}\;.
\end{equation}
This means that, if a reverse data-processing theorem holds, it must be formulated differently.

\subsection{Replacing $H$ with $\hmin$}

Even though we know that ``less noisy'' does not imply ``degradable,'' in what follows we show that just a slight formal modification in the definition of ``less noisy, '' Eq.~(\ref{eq:less-noisy}), is enough to obtain the sought-after reverse data-processing theorem. Such a slight modification consists in replacing, in point~(i) of Definition~\ref{def:less-noisy}, the Shannon conditional entropy $H(\cdot|\cdot)$ with the conditional min-entropy $\hmin(\cdot|\cdot)$.

\begin{theorem}\label{theo:classical-reverse}
	\begin{svgraybox}
	Given are two noisy channels $w:\set{X}\to\set{Y}$ and $w':\set{X}\to\set{Z}$. The following are equivalent:
	\begin{enumerate}[label=\roman*)]
		\item for any set $\set{U}$ and for any initial joint distribution $p(x,u)$, the resulting distributions $\sum_x w(y|x)p(x,u)$ and $\sum_x w'(z|x)p(x,u)$ always satisfy the inequality
		\begin{equation}\label{eq:class-min-ineq}
			\hmin(U|Y)\le\hmin(U|Z)\;;
		\end{equation}
		\item $w$ is degradable into $w'$, namely, there exists another channel $\varphi:\set{Y}\to\set{Z}$ such that $w'(z|x)=\sum_y\varphi(z|y)w(y|x)\;$.
	\end{enumerate}
	\end{svgraybox}
\end{theorem}	

\begin{proof}
	(ii)$\implies$(i) is a direct consequence of the data-processing inequality for $\hmin$. Suppose that there exists another conditional probability distribution $\varphi(z|y)$ such that $w'(z|x)=\sum_y\varphi(z|y)w(y|x)$. This means that the random variable $Z$ is obtained locally from $Y$, i.e., the four random variables $(U,X,Y,Z)$ form a Markov chain $U\to X\to Y\to Z$. This implies that~(\ref{eq:class-min-ineq}) holds.
	
	In order to prove (i)$\implies$(ii), let us assume that the inequality in~(\ref{eq:class-min-ineq}) holds for any initial joint distribution $p(x,u)$. Exponentiating both sides, and using Eq.~(\ref{eq:class-hmin}), this is equivalent to
	\begin{equation}
		\pguess(U|Y)\ge\pguess(U|Z)\;,
	\end{equation}
	namely,
	\begin{equation}\label{eq:class-guess-prob}
		\max_{\varphi}\sum_{u,y,x}\varphi(u|y)w(y|x)p(x,u) \ge \max_{\varphi'}\sum_{u,z,x}\varphi'(u|z)w'(z|x)p(x,u)\;,
	\end{equation}
	for all choices of $p(x,u)$. In the above equation, the noisy channels $\varphi$ and $\varphi'$ represent the decision functions that the statistician designs in order to optimally guess the value of $U$.
	
	Let us choose $U$ such that its support coincides with that of $Z$, i.e., $\set{U}\equiv\set{Z}$. We can therefore denote its states by $z'$. Let us also fix the guessing strategy on the right-hand side of~(\ref{eq:class-guess-prob}) to be $\varphi'(z'|z)\equiv\delta_{z',z}$, i.e., 1 if $z'=z$ and 0 otherwise. Then, we know that there exists a decision function  $\varphi(z'|y)$ such that
	\begin{align}
		0&\ge\sum_{z',z,x}\delta_{z',z}w'(z|x)p(x,z')-\sum_{z',y,x}\varphi(z'|y)w(y|x)p(x,z')\\
		&=\sum_{z',x}w'(z'|x)p(x,z')-\sum_{z',y,x}\varphi(z'|y)w(y|x)p(x,z')\\
		&=\sum_{z',x}\left[w'(z'|x)p(x,z')-\sum_y\varphi(z'|y)w(y|x)p(x,z') \right]\\
		&=\sum_{z',x}\left[w'(z'|x)-\sum_y\varphi(z'|y)w(y|x) \right]p(x,z')\;.
	\end{align}
	In other words, for any $p(x,z')$, there exists a $\varphi(z'|y)$ such that the above inequality holds. This is equivalent to say that
	\begin{equation}
		\max_{p}\min_{\varphi}\sum_{z',x}\left[w'(z'|x)-\sum_y\varphi(z'|y)w(y|x) \right]p(x,z')\le0\;.
	\end{equation}
	We now invoke the minimax theorem (in the form reported in Appendix~2, Theorem~\ref{th:minimax}) and exchange the order of the two optimizations:
	\begin{equation}\label{eq:class-minmax}
		\min_{\varphi}\max_{p}\sum_{z',x}\left[w'(z'|x)-\sum_y\varphi(z'|y)w(y|x) \right]p(x,z')\le0\;.
	\end{equation}
	
	Let us now introduce the quantity
	\begin{equation}
		\Delta_\varphi(z',x)\defeq w'(z'|x)-\sum_y\varphi(z'|y)w(y|x)\;.
	\end{equation}
	First of all, we notice that the maximum in Eq.~(\ref{eq:class-minmax}) is reached when the distribution $p(x,z')$ is entirely concentrated on an entry where $\Delta_\varphi(z',x)$ is maximum, that is,
	\begin{align}
		0&\ge\min_{\varphi}\max_{p}\sum_{z',x}\left[w'(z'|x)-\sum_y\varphi(z'|y)w(y|x) \right]p(x,z')\\
		&= \min_{\varphi}\max_{z',x}\Delta_\varphi(z'x)\;.
	\end{align} 
	In general, $\Delta_\varphi(z',x)$ does not have a definite sign, however, since $\sum_{z',x}\Delta_\varphi(z',x)=0$ (as a consequence of the normalization of probabilities), it must be that $\max_{z',x}\Delta_\varphi(z',x)\ge 0$ (otherwise, of course, one would have $\sum_{z',x}\Delta_\varphi(z',x)<0$). The above inequality hence implies that $\min_\varphi\max_{z',x}\Delta_\varphi(z',x)=0$. In turns this implies, again because $\sum_{z',x}\Delta_\varphi(z',x)=0$, that $\Delta_\varphi(z',x)=0$ for all $z'$ and $x$. In other words, we showed that there exists a $\varphi(z'|y)$ such that
	\begin{equation}
		w'(z'|x)=\sum_y\varphi(z'|y)w(y|x),
	\end{equation}
	for all $z',x$, which coincides with the definition of degradability.\qed
\end{proof}

\begin{remark}
	From the proof we see that in point~(ii) of Theorem~\ref{theo:classical-reverse} it is possible to restrict, without loss of generality, the random variable $U$ to be supported on the set $\set{Z}$, i.e., the same supporting the output of $w'$.\qed
\end{remark}

\section{The Fundamental Lemma for Quantum Channels}
\label{sec:lemma}

The following lemma plays a crucial role in the derivation of reverse data-processing theorems valid in the quantum case. 

\begin{lemma}\label{lem:fundamental}
	\begin{svgraybox}
	Let $\Phi_A:\op{L}(\set{H}_A)\to\op{L}(\set{H}_B)$ and $\Phi'_A:\op{L}(\set{H}_A)\to\op{L}(\set{H}_{B'})$ be two quantum channels. For any set $\set{U}=\{u\}$, the following are equivalent:
	\begin{enumerate}[label=\roman*)]
		\item for all ensembles $\{p(u);\omega^u_A\}\;$,
		\begin{equation}\label{eq:quantum-ii}
		\pguess\{p(u);\Phi_A(\omega^u_A)\}\ge\pguess\{p(u);\Phi'_A(\omega^u_A)\}\;;
		\end{equation}
		\item for any POVM $\{Q^u_{B'}\}$, there exists a POVM $\{P^u_B\}$ such that
			\begin{equation}\label{eq:quantum-iv}
			\Tr{\Phi'_A(\omega_A)\ Q^u_{B'}}=\Tr{\Phi_A(\omega_A)\ P^u_B}\;,
			\end{equation}
		for all $u\in\set{U}$ and all $\omega_A\in\op{D}(\sH_A)$.
	\end{enumerate}
	\end{svgraybox}
\end{lemma}

\begin{proof}
	The fact that (ii) implies (i) follows by definition of guessing probability. We therefore prove the converse, namely, that (i) implies (ii).
	
	Let us rewrite condition~(\ref{eq:quantum-ii}) explicitly as follows: for all ensembles $\{p(u);\omega^u_A \}\;$,
	\begin{equation}\label{eq:cond-start}
		\max_{P}\sum_up(u)\Tr{\Phi_A(\omega^u_A)\ P^u_B}\ge\max_{Q}\sum_up(u)\Tr{\Phi'_A(\omega^u_A)\ Q^u_{B'}}\;,
	\end{equation}
	where the maxima are taken over all possible POVMs. Introduce now an auxiliary Hilbert space $\set{H}_R\cong\set{H}_{A}$, and denote by $\phi^+_{RA}$ a fixed maximally entangled in $\op{D}(\set{H}_R\otimes\set{H}_A)$. Construct then the Choi operators corresponding to channels $\Phi$ and $\Phi'$, namely,
	\begin{equation}\label{eq:precedes-explanation}
		\chi_{RB}\defeq(\id_R\otimes\Phi_A)\phi^+_{RA}\qquad\text{and}\qquad\chi'_{RB'}\defeq(\id_R\otimes\Phi'_A)\phi^+_{RA}\;.
	\end{equation}
	
	Noticing that, for any ensemble $\{p(u);\omega^u_A\}$ with $\sum_up(u)\omega^u_A= I_A/d_A$, there exists a POVM $\{E^u_R\}$ such that $p(u)\omega^u_A=\PTr{R}{\phi^+_{RA}\ (E^u_R\otimes I_A)}$, we immediately see that, if condition~(\ref{eq:cond-start}) above holds, then, for any POVM $\{E^u_R\}$,
	\begin{equation}\label{eq:povm-pre-obs}
		\max_P\sum_u\Tr{\chi_{RB}\ (E^u_R\otimes P^u_{B})}\ge\max_Q\sum_u\Tr{\chi'_{RB'}\ (E^u_R\otimes Q^u_{B'})}\;.
	\end{equation}
	
	We now prove that condition~(\ref{eq:povm-pre-obs}) above in turns implies that, for any collection of Hermitian operators $\{O^u_R \}$,
	\begin{equation}\label{eq:cond-obs}
		\max_P\sum_u\Tr{\chi_{RB}\ (O^u_R\otimes P^u_{B})}\ge\max_Q\sum_u\Tr{\chi'_{RB'}\ (O^u_R\otimes Q^u_{B'})}\;.
	\end{equation}
	The crucial observation here is that, given a collection of Hermitian operators $\{O^u_R\}\;$, we can always derive from it a POVM $\{E^u_R\}$ given by
	\begin{equation}
	E^u_R\defeq\frac{1}{\alpha|\set{U}|}\left\{O^u_R+\alpha I_R-\frac{1}{|\set{U}|}\Sigma_R \right\}\;,
	\end{equation}
	with $\Sigma_R\defeq\sum_uO^u_R$ and $\alpha>0$ sufficiently large so that $O^u_R+\alpha I_R-|\set{U}|^{-1}\Sigma_R$ is nonnegative for all $u$. Therefore, assuming that inequality~(\ref{eq:povm-pre-obs}) holds for any POVM $\{E^u_R \}$, we have that
	\begin{align}
		&\max_P\sum_u\Tr{\chi_{RB}\ (O^u_R\otimes P^u_{B})}\nonumber\\
		&=\alpha|\set{U}|\max_P\sum_u\Tr{\chi_{RB}\ (E^u_R\otimes P^u_{B})}-\alpha\Tr{\chi_{RB}}+\frac{1}{|\set{U}|}\Tr{\chi_{RB}\ (\Sigma_R\otimes I_B)}\nonumber\\
		&=\alpha|\set{U}|\max_P\sum_u\Tr{\chi_{RB}\ (E^u_R\otimes P^u_{B})}-\alpha+\frac{1}{|\set{U}|}\Tr{\PTr{B}{\chi_{RB}}\ \Sigma_R}\nonumber\\
		&\ge \alpha|\set{U}|\max_Q\sum_u\Tr{\chi'_{RB'}\ (E^u_R\otimes Q^u_{B'})}-\alpha+\frac{1}{|\set{U}|}\Tr{\PTr{B'}{\chi'_{RB'}}\ \Sigma_R}\label{eq:cruc-ineq}\\
		&=\max_Q\sum_u\Tr{\chi'_{RB'}\ (O^u_R\otimes Q^u_{B'})},\nonumber
	\end{align}
	for any collection of Hermitian operators $\{O^u_R\}\;$. Inequality~(\ref{eq:cruc-ineq}) above is a consequence of condition~(\ref{eq:povm-pre-obs}) together with the identity $\PTr{B}{\chi_{RB}}=\PTr{B'}{\chi'_{RB'}}=I_R/d_R$. Hence we showed that condition~(\ref{eq:cond-obs}) holds if condition~(\ref{eq:povm-pre-obs}) holds, even though the former looks at first sight more general than the latter. The vice versa is true simply because any POVM is, in particular, a family of Hermitian operators.
	
	Let us now denote by $\set{L}(\set{U})$ the set of operator tuples
	\begin{equation}
		\vec{a}\equiv\left(a^u:u\in\set{U} \right)\;,\qquad a^u\in\op{L}_H(\sH_R)\;,
	\end{equation}
	with inner product
	\begin{align}
		\vec{a}\cdot\vec{b}\defeq\sum_u\Tr{a^u b^u}\;.
	\end{align}
	We then define $\set{C}(\chi;\set{U})$ as the convex subset of $\set{L}(\set{U})$ containing tuples $\vec{b}$ such that $b^u\defeq\PTr{B}{\chi_{RB}\ (I_R\otimes P^u_B)}$, for varying POVM $\{P^u_B\}$. [The fact that $\set{C}(\chi;\set{U})$ is convex is a direct consequence of the fact that the set of POVMs supported on $\set{U}$ is convex.] In the same way, we also define $\set{C}'(\chi';\set{U})$. For the sake of simplicity of notation, when no confusion arises, we simply denote $\set{C}(\chi;\set{U})$ as $\set{C}$ and $\set{C}'(\chi';\set{U})$ as $\set{C}'$. Using this notation, condition~(\ref{eq:cond-obs}) becomes
	\begin{equation}
		\max_{\vec{b}\in\set{C}}\vec{a}\cdot\vec{b}\ge\max_{\vec{b}'\in\set{C}'}\vec{a}\cdot\vec{b}'\;,
	\end{equation}
	for all $\vec{a}\in\set{L}(\set{U})$. [Here $a^u=O^u_R$.]
	
	Hence, we turned the initial conditions involving guessing probabilities into a family of linear constraints on two convex sets, $\set{C}$ and $\set{C}'$. Then, a direct application of the separation theorem for convex sets (see Corollary~\ref{theo:sep2} in Appendix~3), leads us to conclude that
	\begin{equation}
		\set{C}(\chi;\set{U})\supseteq\set{C}'(\chi';\set{U})\;.
	\end{equation}
	
	In other words, condition~(\ref{eq:quantum-ii}) in the statement of the lemma implies that, for any POVM $\{Q^u_{B'}\}$, there exists a POVM $\{P^u_B\}$ such that
	\begin{equation}\label{eq:almost-final-lemma}
		\PTr{B}{\chi_{RB}\ (I_R\otimes P^u_B)}=\PTr{B'}{\chi'_{RB'}\ (I_R\otimes Q^u_{B'})}\;,
	\end{equation} 
	for all $u\in\set{U}$.
	
	The final step consists in noticing that any state $\omega_A$ can be written as $\PTr{R}{\phi^+_{RA}\ (E_R\otimes I_A)}$ for some $E_R\in\op{L}_+(\sH_R)$. Therefore, multiplying both sides of~(\ref{eq:almost-final-lemma}) by $E_R$ and taking the trace, we obtain
	\begin{align}
		\Tr{\Phi_A(\omega_A)\ P^u_B}&=\Tr{\chi_{RB}\ (E_R\otimes P^u_B)}\\
		&=\Tr{\chi'_{RB'}\ (E_R\otimes Q^u_{B'})}\\
		&=\Tr{\Phi'_A(\omega_A)\ Q^u_{B'}}\;,
	\end{align}
	which of course holds for any choice of $E_R$, that is, $\omega_A$, as claimed.\qed
\end{proof}

\begin{remark}
	As explained in the paragraph following Eq.~(\ref{eq:precedes-explanation}), the above proof shows that, in particular, the ensembles $\{p(u);\omega^u_A \}$ in point~(i) can be restricted, without loss of generality, to ensembles with maximally mixed average, i.e., $\sum_up(u)\omega^u_A\propto I_A$.\qed
\end{remark}	

\begin{remark}\label{rem:stat-morph}
	We notice that point~(ii) can be alternatively formulated as follows: for any POVM $\{Q^u_{B'}\}$, there exists a POVM $\{P^u_B\}$ such that
	\begin{equation}\label{eq:quantum-iii}
	\left(\Phi'\right)^\dag\left(Q^u_{B'}\right)=\Phi^\dag\left(P^u_B\right)\;,
	\end{equation}
	for all $u\in\set{U}$, where $\Phi^\dagger$ denotes the trace-dual defined in Eq.~(\ref{eq:trace-dual}). \qed
\end{remark}

\subsection{Quantum Statistical Morphisms}

Let us now choose the set $\set{U}$ in Lemma~\ref{lem:fundamental} so that its size $|\set{U}|$ is equal to $(\text{dim}\sH_{B'})^2$. Assuming that channels $\Phi$ and $\Phi'$ actually satisfy either~(\ref{eq:quantum-ii}) or~(\ref{eq:quantum-iv}), let us set the POVM $\{Q^u_{B'}\}$ to  be informationally complete, that is, $\operatorname{span}\{Q^y_{B'}\}=\op{L}(\sH_{B'})$. Then, if $\{P^u_B\}$ is any POVM satisfying the equality~(\ref{eq:quantum-iii}) in Remark~\ref{rem:stat-morph}, the relation
\begin{equation}
	Q^y_{B'}\longmapsto P^y_B\;,\qquad y\in\set{Y}\;,
\end{equation}
can be used to define a linear map $\Gamma:\op{L}(\sH_{B})\to\op{L}(\sH_{B'})$ with the following properties:
\begin{enumerate}
	\item let $\{\Xi^y_{B'}\}$ be the unique dual of $\{Q^y_{B'}\}$, in the sense that $X_{B'}=\sum_y\Tr{Q^y_{B'}\ X_{B'}}\Xi^y_{B'}\;$, for all $X_{B'}\in\op{L}(\sH_{B'})\;$; then the action of $\Gamma$ is given by $\Gamma(\cdot)=\sum_y\Tr{P^y_{B}\ \cdot}\;\Xi^y_{B'}\;$;
	\item $\Gamma$ is Hermiticity-preserving, i.e., $X=X^\dag$ implies that $\Gamma(X)=\left[ \Gamma(X) \right]^\dag\;$;
	\item $\Gamma$ is trace-preserving;
	\item $\Phi'=\Gamma\circ\Phi\;$.
\end{enumerate}
In particular, the map $\Gamma$, as defined above, is positive and trace-preserving on the output (meant as the whole linear range) of $\Phi$. In order to prove this, let $X_A\in\op{L}(\sH_A)$ be any operator such that $\Phi_A(X_A)\ge0$. (Notice that $X_A$ need not be positive itself.) Then $\Gamma_B(\Phi_A(X_A))\ge 0$. This is because $\Gamma\circ\Phi=\Phi'$ and we know, from Eq.~(\ref{eq:quantum-iv}), that for any positive operator $Q_{B'}$ there exists a positive operator $P_B$ such that $\Tr{Q_{B'}\ \Gamma_B(\Phi_A(X_A))}=\Tr{Q_{B'}\ \Phi'_A(X_A)}=\Tr{P_B\ \Phi_A(X_A)}$. Hence, we know that for any positive operator $Q_{B'}$, $\Tr{Q_{B'}\ \Gamma_B(\Phi_A(X_A))}\ge0$ whenever $\Phi_A(X_A)\ge 0$, which is the definition of positivity.

Following the terminology of~\cite{morse,cencov}, the following definition was introduced in~\cite{buscemi-qblackwell}:

\begin{definition}
	Given a channel $\Phi:\op{L}(\sH_A)\to\op{L}(\sH_B)$, a linear map $\Gamma:\op{L}(\sH_B)\to\op{L}(\sH_C)$ is said to be a \textit{quantum statistical morphism} of $\Phi$ if and only if, for any state $\omega_A$ and any POVM $\{Q^y_C\}$, there exists a POVM $\{P^y_B\}$ such that
	\begin{equation}\label{eq:stat-morph-cond}
		\Tr{(\Gamma_B\circ\Phi_A)(\omega_A)\ Q^y_C}=\Tr{\Phi_A(\omega_A)\ P^y_B}\;,
	\end{equation}
	for all $y$. \qed
\end{definition}

It is easy to verify that an everywhere positive trace-preserving linear map is always a statistical morphisms for \textit{any channel}, as long as the composition between the two is well defined. Then, the natural question is whether a linear map defined as $\Gamma$ above can always be extended to become positive and trace-preserving  \textit{everywhere}, not only on the range of $\Phi$. The question was answered in the negative by Matsumoto, who gave an explicit counterexample in Ref.~\cite{matsumoto_example}.

Vice versa, one may ask whether any linear map that is positive and trace-preserving on the range of $\Phi$ is a well-defined statistical morphism of $\Phi$ or not. Also in this case, the answer is in the negative: the fact that condition~(\ref{eq:stat-morph-cond}) must hold \textit{for any} POVM (in particular, for any number of outcomes) is strictly stronger than just positivity, for which is enough if condition~(\ref{eq:stat-morph-cond}) holds \textit{only for} two-outcome POVMs.

Statistical morphisms hence lie somewhere in between linear maps that are positive and trace-preserving (PTP) everywhere, and those that are so only on the range of $\Phi$:
\begin{equation}
	\text{PTP everywhere}\underset{\centernot\impliedby}{\implies}\textrm{stat. morph. of $\Phi$}\underset{\centernot\impliedby}{\implies}\textrm{PTP on $\operatorname{range}(\Phi)$}\;.
\end{equation}

We summarize the contents of this section in one definition and one corollary.

\begin{definition}
	Given are two quantum channels $\Phi:\op{L}(\sH_A)\to\op{L}(\sH_B)$ and $\Phi':\op{L}(\sH_A)\to\op{L}(\sH_{B'})$. For a given set $\set{U}$, we say that $\Phi$ is $\set{U}$-\textit{sufficient} for $\Phi'$, in formula,
	\begin{equation}
		\Phi\succeq_\set{U}\Phi'\;,
	\end{equation}
	 if and only if either of the conditions in Lemma~\ref{lem:fundamental} hold.\qed
\end{definition}

\begin{corollary}\label{coro:q-stat-morph}
	\begin{svgraybox}
	Given are two quantum channels $\Phi:\op{L}(\sH_A)\to\op{L}(\sH_B)$ and $\Phi':\op{L}(\sH_A)\to\op{L}(\sH_{B'})$. The following are equivalent:
	\begin{enumerate}[label=\roman*)]
		\item $\Phi\succeq_\set{U}\Phi'$, for any set $\set{U}\;$;
		\item there exists a quantum statistical morphism $\Gamma:\op{L}(\sH_B)\to\op{L}(\sH_{B'})$ of $\Phi$ such that $\Phi'=\Gamma\circ\Phi\;$.
	\end{enumerate}
	\end{svgraybox}
\end{corollary}

\begin{remark}\label{rem:analogy}
	Using the correspondence between ensembles and bipartite states, together with the relation between guessing probability and conditional min-entropy, given in Appendix~1 in Eqs.~(\ref{eq:cq-state}) and~(\ref{eq:quantum-hmin}), we notice that the condition $\Phi\succeq_\set{U}\Phi'$ can be equivalently written as
	\begin{equation}
		\hmin(U|B)\le\hmin(U|B'),
	\end{equation}
	where the entropies are computed with respect to states $(\id_U\otimes\Phi_A)(\omega_{UA})$ and $(\id_U\otimes\Phi'_A)(\omega_{UA})$, respectively. This is equivalent to the formulation used in Theorem~\ref{theo:classical-reverse}.\qed
\end{remark}

\section{A Semiclassical (Semiquantum) Reverse-Data Processing Theorem}
\label{sec:semiclassical}

We consider in this section the case in which the output of a quantum channel is classical, in the precise sense that the range is supported on a commutative subalgebra.

\begin{theorem}\label{theo:povms}
	\begin{svgraybox}
	Given are two quantum channels $\Phi:\op{L}(\sH_A)\to\op{L}(\sH_B)$ and $\Phi':\op{L}(\sH_A)\to\op{L}(\sH_{B'})$. Assuming that the output of $\Phi'$ is classical, i.e.,
	\begin{equation}
		[\Phi'(X),\Phi'(Y)]=0,\qquad\forall X,Y\in\op{L}(\sH_A)\;,
	\end{equation}
	the following are equivalent:
	\begin{enumerate}[label=\roman*)]
		\item $\Phi\succeq_\set{U}\Phi'$, for any set $\set{U}\;$;
		\item $\Phi\succeq_\set{U}\Phi'$, for a set $\set{U}$ such that $|\set{U}|=\dim\sH_{B'}\;$;
		\item there exists a quantum channel $\Psi:\op{L}(\sH_B)\to\op{L}(\sH_{B'})$ such that $\Phi'=\Psi\circ\Phi\;$.
	\end{enumerate}
	\end{svgraybox}
\end{theorem}

\begin{proof}
	Since the implications (iii)$\implies$(i)$\implies$(ii) are either trivial of direct consequence of the data-processing inequality for the guessing probability, we only prove the implication (ii)$\implies$(iii).

	 Since $|\set{U}|=\dim\sH_{B'}$, we can use the elements $u\in\set{U}$ to label an orthonormal basis $\{\ket{u}:u\in\set{U}\}$ of $\sH_{B'}$. Assuming~(ii), we know from Lemma~\ref{lem:fundamental} that~(\ref{eq:quantum-iv}) holds, so, in particular, we know that there exists a POVM $\{P^u_B\}$ such that
	 \begin{equation}\label{eq:qc-equiv}
		 \Tr{\Phi'_A(\omega_A)\ \flag{u}_{B'}}=\Tr{\Phi(\omega_A)\ P^u_B}\;,
	 \end{equation}
	 for all $u$ and all $\omega_A\in\op{D}(\sH_A)$.
	 
	 We now use the fact that the output of $\Phi'$ is classical and assume that any operator in the range of $\Phi'$ can be diagonalized on the basis $\{\ket{u} \}$. This means that
	 \begin{equation}
		 \Phi'_A(\cdot)=\sum_{u\in\set{U}}\Tr{\Phi'_A(\cdot)\ \flag{u}_{B'}}\flag{u}_{B'}\;.
	 \end{equation}
	 Using Eq.~(\ref{eq:qc-equiv}), and defining a measure-and-prepare channel $\Psi:\op{L}(\sH_B)\to\op{L}(\sH_{B'})$ by the relation
	 \begin{equation}
		 \Psi(\cdot)\defeq\sum_u\Tr{\cdot\ P^u_B }\flag{u}_{B'}\;,
	 \end{equation}
	 we finally have that $\Phi'=\Psi\circ\Phi$.\qed
\end{proof}

\begin{remark}\label{rem:analogy2}
	In order to highlight the perfect analogy with Theorem~\ref{theo:classical-reverse}, we recall that the relation between guessing probability and conditional min-entropy (see Appendix~1) allows us to rewrite points~(i) and~(ii) of Theorem~\ref{theo:povms} as:
	\begin{equation}
		\hmin(U|B)\le\hmin(U|B')\;.
	\end{equation}
	See also Remark~\ref{rem:analogy} above. \qed
\end{remark}

\begin{remark}
	It is possible to show that Theorem~\ref{theo:classical-reverse} becomes a corollary of Theorem~\ref{theo:povms}. Consider in fact the situation in which both $\Phi$ and $\Phi'$ are \textit{classical-quantum channels}, namely, $\Phi:\set{X}\to\op{L}(\sH_B)$ and $\Phi':\set{X}\to\op{L}(\sH_{B'})$, with $\Phi(x)\defeq\rho^x_B\in\op{D}(\sH_B)$ and $\Phi'(x)\defeq\sigma^x_{B'}\in\op{D}(\sH_{B'})$. Assume moreover that $[\rho^x,\rho^{x'}]=0$ and $[\sigma^x,\sigma^{x'}]=0$, for all $x,x'\in\set{X}$. We are hence in a scenario much more restricted than that of Theorem~\ref{theo:povms}: in fact, by identifying commuting states with the probability distributions of their eigenvalues, we recover the classical framework and the statement of Theorem~\ref{theo:classical-reverse}.\qed
\end{remark}

\begin{remark}
	Theorem~\ref{theo:classical-reverse}, the classical reverse data-processing inequality, has thus two different proofs: one using the minimax theorem and another using the separation theorem for convex sets. Despite the fact that minimax theorem and separation theorem are ultimately equivalent~\cite{minimax}, the minimax theorem allows for an easier treatment of the approximate case, which is a very relevant point but goes beyond the scope of the present contribution. The interested reader may refer to Refs.~\cite{buscemi-prob-inf-trans,jencova-approx}. \qed
\end{remark}

\section{A Fully Quantum Reverse Data-Processing Theorem}
\label{sec:quantum}

We consider in this section the case of two completely general quantum channels, with the only restriction that the input space is the same for both.

\begin{theorem}\label{theo:quantum}
	\begin{svgraybox}
		Given are two quantum channels, $\Phi:\op{L}(\sH_A)\to\op{L}(\sH_B)$ and $\Phi':\op{L}(\sH_A)\to\op{L}(\sH_{B'})$, and an auxiliary Hilbert space $\sH_{B''}\cong\sH_{B'}$. The following are equivalent:
		\begin{enumerate}[label=\roman*)]
			\item $\id_{B''}\otimes\Phi_A\succeq_\set{U}\id_{B''}\otimes\Phi'_A$, for any set $\set{U}\;$;
			\item $\id_{B''}\otimes\Phi_A\succeq_\set{U}\id_{B''}\otimes\Phi'_A$, for a set $\set{U}$ such that $|\set{U}|=\dim(\sH_{B''}\otimes\sH_{B'})=(\dim\sH_{B'})^2\;$;
			\item there exists a quantum channel $\Psi:\op{L}(\sH_B)\to\op{L}(\sH_{B'})$ such that $\Phi'=\Psi\circ\Phi\;$.
		\end{enumerate}
	\end{svgraybox}
\end{theorem}

\begin{remark}
	In terms of the conditional min-entropy, points~(i) and~(ii) above can be written as
	\begin{equation}
		\hmin(U|B''B)\le\hmin(U|B''B')\;,
	\end{equation}
	with obvious meaning of symbols. See also Remarks~\ref{rem:analogy} and~\ref{rem:analogy2} above. \qed
\end{remark}

\begin{proof}
	Since the implication (iii)$\implies$(i)$\implies$(ii) is straightforward, we prove here only that (ii)$\implies$(iii).
	
	Let $\sH_{B'''}$ be a futher auxiliary Hilbert space such that $\sH_{B'''}\cong\sH_{B''}\cong\sH_{B'}$. We begin by showing that, if $\id_{B''}\otimes\Phi_A\succeq_\set{U}\id_{B''}\otimes\Phi'_A$, then, for any POVM $\{Q^u_{B''B'}\}\;$, there exists a POVM $\{P^u_{B''B}\}$ such that
	\begin{equation}\label{eq:inter-to-prove}
	\begin{split}
		&\PTr{B''B'}{(\phi^+_{B'''B''}\otimes\Phi'_A(\cdot))\ (I_{B'''}\otimes Q^u_{B''B'})}\\
		&=\PTr{B''B}{(\phi^+_{B'''B''}\otimes\Phi_A(\cdot))\ (I_{B'''}\otimes P^u_{B''B})}\;,
	\end{split}
	\end{equation}
	where $\phi^+_{B'''B''}$ is a maximally entangled state in $\sH_{B'''}\otimes\sH_{B''}$. In fact, Lemma~\ref{lem:fundamental} states that, for any POVM $\{Q^u_{B''B'}\}$, there exists a POVM $\{P^u_{B''B}\}$ such that
	\begin{align}
		\Tr{(\id_{B''}\otimes\Phi'_A)(\cdot_{B''A})\ Q^u_{B''B'} }=\Tr{(\id_{B''}\otimes\Phi_A)(\cdot_{B''A})\ P^u_{B''B} }\;,
	\end{align}
	for all $u\in\set{U}$. In particular, for any family of states $\{\xi^x_{B''} \}_x$ on $\sH_{B''}$, we have
	\begin{equation}\label{eq:interm1}
		\Tr{(\id_{B''}\otimes\Phi'_A)(\xi^x_{B''}\otimes\cdot_A)\ Q^u_{B''B'} }=\Tr{(\id_{B''}\otimes\Phi_A)(\xi^x_{B''}\otimes\cdot_A)\ P^u_{B''B} }\;,
	\end{equation}
	for all $u$ and all $x$. Let us choose $\xi^x_{B''}=\PTr{B'''}{\phi^+_{B'''B''}\ (\Xi^x_{B'''}\otimes I_{B''})}$ for some complete set of positive operators $\{\Xi^x_{B'''} \}_x$. Hence Eq.~(\ref{eq:interm1}) becomes
	\begin{align}
		&\Tr{(\id_{B'''}\otimes\id_{B''}\otimes\Phi'_A)(\phi^+_{B'''B''}\otimes\cdot_A)\ (\Xi^x_{B'''}\otimes Q^u_{B''B'}) }\\
		&=\Tr{(\id_{B'''}\otimes\id_{B''}\otimes\Phi_A)(\phi^+_{B'''B''}\otimes\cdot_A)\ (\Xi^x_{B'''}\otimes P^u_{B''B}) }\;,
	\end{align}
	for all $u$ and all $x$. But since the family $\{\Xi^x_{B'''}\}_x$ has been chosen to be complete, the above equality implies the equality of the operators in Eq.~(\ref{eq:inter-to-prove}).
	
	Now, we can use generalized teleportation and show that
	\begin{align}
		\Phi'_A(\cdot)=\sum_uW^u_{B'''}\left\{\PTr{B''B'}{(\phi^+_{B'''B''}\otimes\Phi'_A(\cdot))\  (I_{B'''}\otimes\beta^u_{B''B'})  }\right\}(W^u_{B'''})^\dag\;,
	\end{align}
	where $\{\beta^u_{B''B'''}:u\in\set{U} \}$ are the $(\dim\sH_{B'})^2$ projectors onto the Bell states, and $\{W^u_{B'''}:u\in\set{U}\}$ are suitable isometries from $\sH_{B'''}$ to $\sH_{B'}$. But then, using Eq.~(\ref{eq:inter-to-prove}) with $Q^u_{B''B'}=\beta^u_{B''B'}$, we obtain
%
	\begin{align}
		\Phi'_A(\cdot)=\sum_uW^u_{B'''}\left\{\PTr{B''B}{(\phi^+_{B'''B''}\otimes\Phi_A(\cdot))\  (I_{B'''}\otimes P^u_{B''B})  }\right\}(W^u_{B'''})^\dag\;.
	\end{align}
	Hence, defining a quantum channel $\Psi:\op{L}(\sH_B)\to\op{L}(\sH_{B'})$ as
	\begin{equation}
		\Psi(\cdot)\defeq\sum_uW^u_{B'''}\left\{\PTr{B''B}{(\phi^+_{B'''B''}\otimes\cdot)\  (I_{B'''}\otimes P^u_{B''B})  }\right\}(W^u_{B'''})^\dag\;,
	\end{equation}
	we finally have that $\Phi'=\Psi\circ\Phi$, as claimed.\qed
\end{proof}

\begin{remark}
	Theorem~\ref{theo:quantum} holds also if the identity channel $\id_{B''}$ is replaced by a \textit{complete} channel, namely, a channel $\Upsilon:\op{L}(\sH_{B''})\to\op{L}(\sH_{B''})$ that is bijective (in the sense of the linear map): linearly independent inputs are transformed into linearly independent outputs. This is so because linearly independent states $\xi^x_{B''}$ in Eq.~(\ref{eq:interm1}) remain linearly independent after the action of $\Upsilon$. In this way, the proof can continue along the same lines.
	
	We notice, in particular, that a channel can be complete despite being entanglement breaking or measure-and-prepare. This implies that the ensembles used to probe channels $\id_{B''}\otimes\Phi_A$ and $\id_{B''}\otimes\Phi'_A$ can always be chosen, without loss of generality, to comprise separable states only.\qed
\end{remark}

\section{The Computational Second Law: An Analogy}
\label{sec:app}

The aim of this section is to construct an analogy, clarifying and somehow strengthening that given by Cover~\cite{cover_1996}, between data-processing theorems and the second law of thermodynamics. In what follows we abandon a formally rigorous language, preferring instead a generic language better suited to highlight the similarities and differences between thermodynamics and information theory.

Theorems~\ref{theo:classical-reverse},~\ref{theo:povms}, and~\ref{theo:quantum}, a part from the formal complications necessary to describe classical and quantum systems together, have all the same simple interpretation that we summarize in two statements (A) and (B):
\[
\begin{minipage}{0.75\textwidth}
	\textbf{Direct statement}: the information that the signal carries about the message (\textit{any} message) cannot increase along a Markov local process;
\end{minipage}
\tag{A}
\]
and
\[
\begin{minipage}{0.75\textwidth}
\textbf{Reverse statement}: if the information that the signal carries \textit{never} increases along a given process, then such a process admits a Markov local realization.
\end{minipage}
\tag{B}
\]
The direct statement corresponds to Cover's law, as formulated in~\cite{cover_1996} (see the \hyperref[quote]{quotation} at the beginning of this paper). Here ``useful information'' is precisely the information that the signal carries about the message, and it is measured by the conditional min-entropy, which is directly related to the guessing probability. The reverse statement, which is consequence of the reverse data-processing theorems that we proved, corresponds to Lieb's and Yngvason's entropy principle~\cite{lieb-yngv}.

In order to make our discussion more concrete, let us consider a thermodynamical system prepared at time $t_0$ and evolving through successive times $t_1\ge t_0$ and $t_2\ge t_1$, as depicted in Fig.~\ref{fig:process} below. The second law of thermodynamics, in the formulation usually attributed to Clausius, states that the following inequality is necessarily obeyed:
\begin{equation}\label{eq:clausius}
\Delta H\ge\frac{\Delta Q}{T},
\end{equation}
where $\Delta H=H(S_2)-H(S_1)$ is the change in thermodynamical entropy of the system and $\Delta Q$ is the heat \textit{absorbed by} the system\footnote{In the precise sense that $\Delta Q$ is positive if heat is injected into the system and negative if heat is extracted from the system; see, e.g., Ref.~\cite{borgakke}}. The above equation basically says that the only way to decrease the entropy of a system is to extract heat from the system. This implies that, if a system is adiabatically isolated (i.e., no heat is exchanged, only mechanical work), then its entropy cannot decrease. Equivalently stated: a decrease in entropy represents a definite witness of the fact that the system \textit{is not} adiabatically isolated and is dumping heat in the environment.

\begin{figure}[b]
	\centering
	\includegraphics[width=0.7\linewidth]{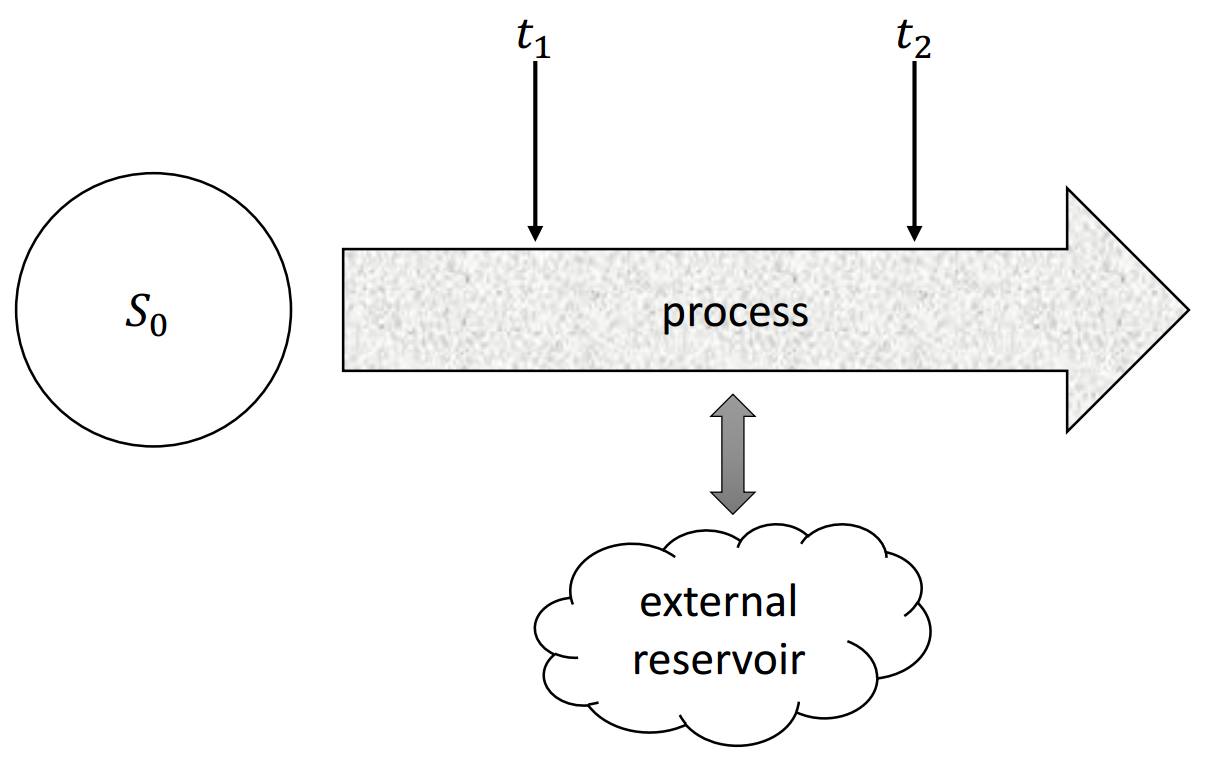}
	\caption{Suppose that a system, prepared at time $t_0$, undergoes a process, and that we observe it at two later times $t_1\ge t_0$ and $t_2\ge t_1$. \textbf{Thermodynamical case}: Clausius' principle and Lieb's and Yngvason's entropy principle state that $\Delta H=H(S_2)-H(S_1)\ge 0$ if and only if the process bringing the system from $t_1$ to $t_2$ can be realized adiabatically (i.e., exchanging only work and no heat). This is equivalent to say that: (i) a decrease in entropy can only be achieved by exchanging heat with an external reservoir; (ii) if the process cannot be realized adiabatically then there is some initial configuration $S_0$ for which a decrease in entropy occurs. \textbf{Information-theoretic case}: the data-processing inequality and the reverse data-processing theorems state that $\Delta \hmin=\hmin(U|S_2)-\hmin(U|S_1)\ge 0$ for all $U$, if and only if the process bringing the system from $t_1$ to $t_2$ is Markov local (i.e., there exists a memoryless channel $\Psi$ such that $S_2=\Psi(S_1)$). This is equivalent to say that: (i) a decrease in the conditional min-entropy can only be achieved in the presence of an external memory storing information about the message and feeding it back into the system at later times; (ii) if the process is not Markov local, then there exists some initial message-signal joint distribution for which a decrease of $\hmin$ occurs.}
	\label{fig:process}
\end{figure}

This part of the second law can be seen as the analogue of statement~(A) above, that is, the usual data-processing inequality. Suppose now that the system $S$ is an information signal. As before, it is prepared at time $t_0$ and then it undergoes a process that is information-theoretic, rather than thermodynamical. If we observe the signal at two times $t_1\ge t_0$ and $t_2\ge t_1$, then we know that if the process is Markov local, then the data-processing inequality holds, namely, the information carried by the signal cannot increase going from $t_1$ to $t_2$. Therefore, any increase in the information carried by the signal is a definite witness of the fact that the process \textit{is not} Markov local, namely, that \textit{an external memory was used as a side resource at some point along the process}.

We now come to the reverse statement~(B), arguing that it is the analogue of Lieb's and Yngvason's entropy principle. The latter states that, assuming the validity of a set of axioms about simple thermodynamical systems\footnote{The most important and debated of which is the \textit{comparability hypothesis}: the interested reader may refer to Uffink~\cite{uffink}.}, a non-decreasing entropy between $t_1$ and $t_2$ is not only necessary (Clausius' principle) but also \textit{sufficient} for the existence of an adiabatic process between the two times. It is clear that the analogy works in this case too: \textit{the reverse data-processing theorems we proved constitute the information-theoretic analogue of Lieb's and Yngvason's entropy principle}. An overview of the analogy is summarized in the table below.
\begin{table}
	\label{tab:1}
	\begin{tabular}{p{0.5\textwidth}p{0.5\textwidth}}
		\hline\noalign{\smallskip}
		\textbf{Thermodynamical setting} & \textbf{Information-theoretic setting} \\
		\noalign{\smallskip}\svhline\noalign{\smallskip}
		thermodynamical system $S$ & message $U$ encoded on signal $S$ \\
		entropy $H(S)$ & conditional min-entropy $\hmin(U|S)$  \\
		Clausius' principle & data-processing inequality  \\
		Lieb--Yngvason entropy principle & reverse data-processing theorem\\
		adiabatically isolated system & computationally isolated system \\
		adiabatic process & Markov local (memoryless) process \\
		heat sink/reservoir & external memory\\
		\noalign{\smallskip}\hline\noalign{\smallskip}
	\end{tabular}
	\caption{Summary of the analogies between the second law of thermodynamics and its computational analogue discussed here.}
\end{table}

Despite the tantalizing analogies, there are however two points (at least) that we should keep in mind before jumping to rash conclusions. The first one is that, while in thermodynamics a process is usually given by an initial state and a final state, in information theory a process is a channel, which acts upon \textit{any input} it receives.

The second point is that the relation presented here between adiabaticity and Markov locality (or memorylessness) has been discussed only on a formal level, but no claim has been made about any \textit{quantitative} relation between the two concepts. However, we would like to conclude this paper speculating about the possibility to envisage a deeper relation between adiabaticity and Markov locality, going beyond the formal analogy presented above. An adiabatically isolated system cannot exchange heat, but can interact with a mechanical device and exchange work with it. Since it is possible to imagine a purely mechanical memory (at least of finite size), it seems that the presence of a memory, in itself, should not violate adiabaticity. But then, a scenario similar to that of Maxwell's demon immediately comes to mind. Indeed, Maxwell's demon violates the second law using nothing but its memory: its actions, included the measurements it performs, are assumed to be otherwise perfectly adiabatic\footnote{Thus the demon can be imagined as a ``perfect clockwork.''}. Hence, it seems that adiabaticity does not play well with the presence of an external memory, even if this is taken to be perfectly mechanical. This fact suggests that adiabaticity and Markov locality may be even closer than what the analogies in Table~1 \textit{prima facie} seem to suggest. This and other questions are left open for future investigations.

\begin{acknowledgement}
	It is a pleasure to thank, in alphabetical order, Ettore Bernardi, Jeremy Butterfield, Weien Chen, Giulio Chiribella, Giacomo Mauro D'Ariano, Nilanjana Datta, G\'abor Hofer-Szab\'o, Koji Maruyama, Keiji Matsumoto, Mil\'an Mosonyi, Masanao Ozawa, Veiko Palge, Paolo Perinotti, David Reeb, and Mark Wilde, whose comments helped in shaping the present ideas at various stages during the past few years.
\end{acknowledgement}

\section*{Appendix 1: Definitions and Notations}
\addcontentsline{toc}{section}{Appendix 1: Definitions and Notations}

Here we review some basic notions and clarify the notation that is used in the paper. The reader familiar with the standard toolbox used in quantum information theory (see, e.g., Ref.~\cite{wilde-qit}) can safely skip to the next section.

All set and spaces considered here are finite or finite dimensional. We denote sets as $\set{X},\set{Y},\set{Z},\set{U},\dots$ and their elements as $x,y,z,u,\dots\;$. Sets support probability distributions, for example, $p(x)$. When we speak of a random variable, for example, $X$, we mean that it is supported by the set $\set{X}$, in the sense that its states are labeled by $x\in\set{X}$, and that each state can occur with probability $p(x)=\operatorname{Pr}\{X=x\}$. When a pair (or a triple etc) of random variables are considered, we write $(X,Y)$ to mean a bipartite random variable supported on the cartesian product $\set{X}\times\set{Y}=\{(x,y):x\in\set{X},y\in\set{Y}\}$ and distributed with joint probability $p(x,y)$. Classical noisy channels are represented by conditional input--output probability distributions $w(y|x)$: in this case we understand that the channel $w$ has input alphabet $\set{X}$ and output alphabet $\set{Y}$ and write $w:\set{X}\to\set{Y}$.

Quantum systems are labeled by $A,B,C,\dots$ and their corresponding finite dimensional Hilbert spaces are denoted as $\sH_A,\sH_B,\sH_C,\dots\;$. The set of linear operators on a Hilbert space $\sH$ is denoted as $\op{L}(\sH)$, the set of Hermitian operators as $\op{L}_H(\sH)$, the set of positive semidefinite operators as $\op{L}_+(\sH)$, and the set of density operators (or states), i.e., positive semidefinite with unit trace, as $\op{D}(\sH)$. Vectors are denoted as kets $\ket{\phi}$, while if we write $\phi$ we mean the corresponding state, that is, the projector $\flag{\phi}$. Given an orthonormal basis $\{\ket{x}:x\in\set{X}\}$ for a Hilbert space, we sometime call the set of orthogonal projectors $\flag{x}$ ``flags,'' since these can be used to model a classical random variable with distinguishable states. For example, given a random variable $X$ with states $x\in\set{X}$ and distribution $p(x)$, we will often think of it as ``embedded'' in a Hilbert space $\sH_X$, with $\dim\sH_X=|\set{X}|$, and described by the state $\sum_xp(x)\flag{x}$. This is a convention commonly used in quantum information theory as it significantly simplifies the analysis of hybrid classical-quantum scenarios.

A family $\{P^x_A:x\in\set{X}\}$ of operators $P^x_A\in\op{L}_+(\sH_A)$ such that $\sum_xP^x_A=I_A$ is called a POVM on $\sH_A$. An ensemble is given by giving a set $\set{X}$, a probability distribution $p(x)$ and a family of states $\rho^x_A\in\op{L}(\sH_A)$: we denote it for brevity as $\{p(x);\rho^x_A \}$, where the set $\set{X}$ is usually understood from the context. Extending the idea mentioned in the preceding paragraph of embedding classical random variables in orthogonal states of a suitable Hilbert space, it is also common to interpret an ensemble as a bipartite state as follows:
\begin{equation}\label{eq:cq-state}
	\{p(x);\rho^x_A \}_{x\in\set{X}}\quad \Longleftrightarrow\quad\rho_{XA}\defeq\sum_{x\in\set{X}}p(x)\flag{x}_X\otimes\rho^x_A\;.
\end{equation}

A linear map $\Phi:\op{L}(\sH_A)\to\op{L}(\sH_B)$ is said to be a \textit{quantum channel} if and only if it is completely positive and trace-preserving. Given a linear map $\Phi:\op{L}(\sH_A)\to\op{L}(\sH_B)$, its \textit{trace-dual} $\Phi^\dag:\op{L}(\sH_B)\to\op{L}(\sH_A)$ is the linear map defined by the relation
\begin{equation}\label{eq:trace-dual}
	 \Tr{X\ \Phi^\dagger(Y)}\defeq\Tr{\Phi(X)\ Y}\;,
\end{equation}
for all $X\in\op{L}(\sH_A)$ and all $Y\in\op{L}(\sH_B)$. $\Phi$ is a channel if and only if $\Phi^\dag$ is completely positive and unit-preserving, i.e., $\Phi^\dagger_B(I_B)=I_A$.

Given a pair of random variables $(X,U)$, the guessing probability of $U$ given $X$ is
\begin{align}
	\pguess(U|X)&\defeq\max_{\varphi}\sum_u\varphi(u|x)p(x,u)\\
	&=\sum_x\max_up(u,x)\;,
\end{align}
where the optimization is done over all channels (decoding strategies) $\varphi:\set{X}\to\set{U}$. In other words, it is the probability of correctly guessing $U$ using the ideal observer decoding strategy on $X$.

The quantum analogue of this is the problem of correctly guessing $U$ given an ensemble of quantum states $\{p(u);\rho^u_A \}\;$. In this case, the role of datum $X$ is played by the quantum system $A$ and the guessing probability is
\begin{equation}
	\pguess(U|A)\defeq\max_P\sum_u\Tr{P^u_A\ \rho^u_A}\;,
\end{equation}
where the optimization is done over all POVMs $\{P^u_A:u\in\set{U}\}$. Notice that in this paper we only consider the case of guessing a classical random variable given a quantum system, so in the expression $\pguess(U|A)$ the roles of $U$ (random variable) and $A$ (quantum system) should always be clearly understandable from the context.

\subsection*{Entropies}
\addcontentsline{toc}{subsection}{Entropies}

The letter $H$ is used to denote the entropy. More precisely, in the case of classical random variables $H(X)\defeq-\sum_xp(x)\log_2p(x)$; in the case of a quantum state $\rho_A$, $H(A)\defeq-\sum_i\lambda_i\log_2\lambda_i$, where the $\lambda$'s are the eigenvalues of $\rho_A$. Following common terminology, the entropy of a classical variable is called the Shannon entropy, while the entropy of a state is called the von Neumann entropy.

Given a pair of random variables $(X,Y)$, the conditional entropy is $H(X|Y)=H(XY)-H(Y)$ and the mutual information is $I(X;Y)=H(X)+H(Y)-H(XY)=H(X)-H(X|Y)$. Given a bipartite state $\rho_{AB}\in\op{D}(\sH_A\otimes\sH_B)$, all the definitions are extended by analogy, for example, $H(A|B)=H(AB)-H(B)$, where $H(AB)$ is the von Neumann entropy of $\rho_{AB}$ and $H(B)$ is the von Neumann entropy of the reduced state $\rho_B=\PTr{A}{\rho_{AB}}$.

von Neumann and Shannon entropies are not the only entropies that are relevant in information theory. Lately, in particular, alternative entropies have been found to play a central role in various information-theoretic scenarios. Such entropies, whose classification is beyond the scope of this work, include for example R\'enyi entropies and, in particular, min- and max-entropies, see, e.g., Ref.~\cite{tomamichel}. The one that is relevant for this work is the so-called \textit{conditional min-entropy} which is given by
\begin{equation}\label{eq:class-hmin}
	\hmin(U|X)=-\log_2\pguess(U|X)
\end{equation}
in the case of two classical random variables, and
\begin{equation}\label{eq:quantum-hmin}
	\hmin(U|A)=-\log_2\pguess(U|A)
\end{equation}
in the case of an ensembles of quantum states. In fact, $\hmin(U|A)$ is the conditional min-entropy of the classical-quantum state $\rho_{XA}$ defined in Eq.~(\ref{eq:cq-state}).

\section*{Appendix 2: The Minimax Theorem}
\addcontentsline{toc}{section}{Appendix 2: The Minimax Theorem}

Here we state a form of the Minimax Theorem as needed in the proof of Theorem~\ref{theo:classical-reverse}, see, e.g., Lemma~4.13 in Ref.~\cite{liese-miescke}:
\begin{theorem}\label{th:minimax}
	Let $\set{S}\subset\mathbb{R}^s$ be a closed convex set and $\set{L}\subset\mathbb{R}^d$ be a polytope. If $f:\set{S}\times\set{L}\to\mathbb{R}$ is continuous and satisfies
	\begin{align}
		&f\left[\alpha y_1+(1-\alpha)y_2,z \right] = \alpha f(y_1,z)+(1-\alpha)f(y_2,z) \label{eq:minimax1}\\
		&f\left[y,\alpha z_1+(1-\alpha)z_2 \right] = \alpha f(y,z_1)+(1-\alpha)f(y,z_2) \label{eq:minimax2}\;,
	\end{align}
	for all $\alpha\in[0,1]$, $y,y_1,y_2\in\set{S}$, and $z,z_1,z_2\in\set{L}$, then
	\begin{equation}
		\max_{z\in\set{L}}\min_{y\in\set{S}}f(y,z)=\min_{y\in\set{S}}\max_{z\in\set{L}}f(y,z)\;.
	\end{equation}
\end{theorem}
In proving Theorem~\ref{theo:classical-reverse} we specialize the above statement to the case in which $\set{S}$ is the set of classical channels $\varphi:Y\to Z$ (indeed convex and closed) and $\set{L}$ is the set of joint probability distributions on $\set{X}\times\set{Z}$ (indeed a polytope). Last thing to check is that conditions~(\ref{eq:minimax1}) and~(\ref{eq:minimax2}) hold: this is a consequence of the fact that the function in the case considered is actually linear in both its variables.

\section*{Appendix 3: The Separation Theorem}
\addcontentsline{toc}{section}{Appendix 3: The Separation Theorem}

Here we give an elementary geometrical proof of the Hahn-Banach separation theorem in its simplest case, i.e. where the sets considered are closed and bounded. For a more general treatment the interested reader may refer to, e.g., Ref.~\cite{rockafellar}.

\begin{theorem}\label{theo:sep}
	Let $C\in\mathbb{R}^n$ be a closed and bounded convex set, and let
	$y\in\mathbb{R}^n$ be a vector that does not belong to $C$,
	i.e. $y\notin C$. Then, there exists a vector $k\in\mathbb{R}^n$ and
	a constant $\alpha\in\mathbb{R}$ such that $k\cdot x<\alpha<k\cdot
	y$, for all $x\in C$. We say that the hyperplane
	$\set{L}:=\{z\in\mathbb{R}^n:z\cdot k=\alpha\}$ separates $C$
	and $y$ strictly.
\end{theorem}

\begin{proof}
	Let $x_0\in C$ be a point such that
	\begin{equation}
	\N{x_0-y}=\min_{x\in C}\N{x-y}>0.
	\end{equation}
	Its existence is guaranteed by the Weierstrass' extreme value
	theorem. The strict inequality comes from the fact that $y\notin C$,
	by assumption.
	
	Let us now define
	\begin{equation}
	k:=y-x_0
	\end{equation}
	and
	\begin{equation}
	\alpha:=\frac 12(k\cdot x_0+k\cdot y)=\frac 12(y\cdot y-x_0\cdot x_0).
	\end{equation}
	
	We note now that
	\begin{equation}
	\begin{split}
	k\cdot y&=(y-x_0)\cdot y\\
	&=\frac 12\left\{(y-x_0)\cdot y+(y-x_0)\cdot y\right\}\\
	&=\frac 12\left\{(y-x_0)\cdot(y-x_0+x_0)+(y-x_0)\cdot y\right\}\\
	&=\frac 12\left\{(y-x_0)x_0+(y-x_0)\cdot y\right\}+\frac
	12(y-x_0)\cdot(y-x_0)\\
	&>\alpha,
	\end{split}
	\end{equation}
	and that
	\begin{equation}
	\begin{split}
	k\cdot x_0&=(y-x_0)\cdot x_0\\
	&=\frac 12\left\{(y-x_0)\cdot x_0+(y-x_0)\cdot x_0\right\}\\
	&=\frac 12\left\{(y-x_0)\cdot(x_0+y-y)+(y-x_0)\cdot x_0\right\}\\
	&=\frac 12\left\{(y-x_0)x_0+(y-x_0)\cdot y\right\}+\frac
	12(y-x_0)\cdot(x_0-y)\\
	&=\frac 12\left\{(y-x_0)x_0+(y-x_0)\cdot y\right\}-\frac
	12(y-x_0)\cdot(y-x_0)\\
	&<\alpha,
	\end{split}
	\end{equation}
	
	Now, let us consider any $x\in C$. By convexity, $(1-p)x_0+px\in C$,
	for any $p\in[0,1]$. Then, 
	we have that
	\begin{equation}
	\begin{split}
	\N{x_0-y}^2&=\min_{x\in C}\N{x-y}^2\\
	&\le\N{(1-p)x_0+px-y}^2\\
	&=\N{(1-p)(x_0-y)+p(x-y)}^2\\
	&=(1-p)^2\N{x_0-y}^2+2p(1-p)(x_0-y)\cdot(x-y)+p^2\N{x-y}^2,
	\end{split}
	\end{equation}
	where we used the formula $\N{w_0+w_1}^2=
	\N{w_0}^2+\N{w_1}^2+2w_0\cdot w_1$, valid for all
	$w_0,w_1\in\mathbb{R}^n$. Therefore,
	\begin{equation}
	0\le p(p-2)\N{x_0-y}^2+2p(1-p)(x_0-y)\cdot(x-y)+p^2\N{x-y}^2.
	\end{equation}
	Let us now consider the case $p\neq 0$. Then,
	\begin{equation}
	0\le (p-2)\N{x_0-y}^2+2(1-p)(x_0-y)\cdot(x-y)+p\N{x-y}^2,
	\end{equation}
	and, taking the limit for $p\to 0$, we finally obtain
	\begin{equation}
	\begin{split}
	0&\le -2 \N{x_0-y}^2+2(x_0-y)\cdot(x-y)\\
	&\le 2\left\{(x_0-y)\cdot(x-y)-(x_0-y)\cdot(x_0-y)\right\}\\
	&=2\left\{(x_0-y)\cdot x-(x_0-y)\cdot x_0\right\}\\
	&=2\left\{-k\cdot x+k\cdot x_0\right\},
	\end{split}
	\end{equation}
	which implies that $k\cdot x\le k\cdot x_0<\alpha<k\cdot y$, for any
	$x\in C$, as claimed.\qed
\end{proof}

For our purpose the following reformulation of Theorem~\ref{theo:sep} is particularly useful:

\begin{corollary}\label{theo:sep2}
	Let $C_1$ and $C_2$ be two closed and bounded convex sets in
	$\mathbb{R}^n$. Then, $C_1\supseteq C_2$ if and only, for every
	vector $k\in\mathbb{R}^n$,
	\begin{equation}
	\max_{x\in C_1}k\cdot x\ge\max_{y\in C_2}k\cdot y.
	\end{equation}
\end{corollary}

\end{document}